\documentclass[10pt]{article}  

\usepackage{amssymb}
\usepackage{amsmath}
\usepackage{amsthm}
\usepackage{apacite}

\usepackage{ctable}

\newtheorem{theorem}{{\bf Theorem}}
\newtheorem{lemma}{{\bf Lemma}}

\newtheorem{example}{{\bf Example}}
\newtheorem{definition}{{\bf Definition}}

\title{Monotonicity axioms in approval-based multi-winner voting rules}

\author{Luis S\'anchez-Fern\'andez$^1$, Jes\'us A. Fisteus$^1$
\\
$^1$Universidad Carlos III de Madrid, Spain \\
luiss@it.uc3m.es,
jaf@it.uc3m.es
}

\begin{document}

\maketitle

\begin{abstract}

In this paper we study several monotonicity axioms in approval-based
multi-winner voting rules. We consider monotonicity with respect to
the support received by the winners and also monotonicity in the size
of the committee. Monotonicity with respect to the support is studied
when the set of voters does not change and when new voters enter the
election. For each of these two cases we consider a strong and a weak
version of the axiom. We observe certain incompatibilities between the
monotonicity axioms and well-known representation axioms
(extended/proportional justified representation) for the voting rules
that we analyze, and provide formal proofs of incompatibility between
some monotonicity axioms and perfect representation. 

\end{abstract}

\section{Introduction}

There are many situations in which it is necessary to aggregate the
preferences of a group of agents to select a finite set of
alternatives. Typical examples are the election of representatives in
indirect democracy, shortlisting candidates for a
position~\cite{elkind:scw17,barbera2008choose}, selection by a company
of the group of products that it is going to offer to its
customers~\cite{lu2011budgeted}, selection of the web pages that
should be shown to a user in response to a given
query~\cite{dwork2001rank,2016arXiv161201434S}, peer grading in
Massive Open Online Courses (MOOCs)~\cite{comSoc:peerGrading} or
recommender systems~\cite{elkind:scw17,naamani2014preference}. The
typical mechanism for such preference aggregations is the use of
multi-winner voting rules.

The use of axioms for analyzing voting rules is well established in
social choice and dates back to the work of
\citeauthor{arrow2012social}~\shortcite{arrow2012social}.  However,
multi-winner voting rules have not been studied much so far from an
axiomatic perspective. In particular, we can cite the work of
\citeauthor{dummet}~\shortcite{dummet},
\citeauthor{elkind:scw17}~\shortcite{elkind:scw17},
\citeauthor{Faliszewski:2019:CSR}~\shortcite{Faliszewski:2019:CSR},
and \citeauthor{woodall:prop}~\shortcite{woodall:prop} for
multi-winner elections that use ranked ballots.  For approval-based
multi-winner elections the concept of representation has been recently
axiomatized by \citeauthor{aziz:scw}~\shortcite{aziz:scw}, who
proposed two axioms called justified representation and extended
justified representation, and
\citeauthor{pjr-aaai}~\shortcite{pjr-aaai}, who proposed a weakening
of extended justified representation that they called proportional
justified representation.

In this paper we complement these previous works with the study of
monotonicity axioms for approval-based multi-winner voting rules.
First of all, we consider monotonicity in the support received by the
winners. Informally, the idea of monotonicity in the support is that
if a subset of the winners in an election sees their support increased
and the support of all the other candidates remains the same, then it
seems reasonable that such candidates should remain in the set of
winners.  Monotonicity with respect to the support is studied when the
set of voters does not change and when new voters enter the
election. Our first contribution is to propose an axiom for each of
these two cases and, for each of these two axioms, to define a strong
and a weak version of the axiom.  We also consider monotonicity in the
size of the committee, although in this case we will reuse an axiom
that has already been proposed by
\citeauthor{elkind:scw17}~\shortcite{elkind:scw17}.

Following the work of
\citeauthor{elkind:scw17}~\shortcite{elkind:scw17} and
\citeauthor{faliszewski2017multiwinner}~\shortcite{faliszewski2017multiwinner}
we will discuss the relevance of these axioms in three different types of
scenarios:

\begin{itemize}

\item

{\bf Excellence}. The goal is to select the best $k$ candidates for a
given purpose. It is supposed that in a second step (out of the scope
of the multi-winner election) one of the selected candidates is finally
selected. Examples of this type of elections are choosing the
finalists of a competition or shortlisting of candidates for a position.

\item

{\bf Diversity}. In this case the goal is that as many voters as
possible have one of their preferred candidates in the
committee. Several examples of this type are discussed by
\citeauthor{elkind:scw17}~\shortcite{elkind:scw17} and
\citeauthor{faliszewski2017multiwinner}~\shortcite{faliszewski2017multiwinner}.
One such example is to select the set of movies that are going to be
offered to the passengers during an air flight (the airline company is
interested in that all passengers find something that they like).

\item

{\bf Proportional representation}. In this case the goal is to select
a committee that represents as precisely as possible the opinions of
the society. The typical example of this scenario are parliamentary
elections.  
  
\end{itemize}

Then, we analyze several well-known voting rules with these axioms. We
observe certain incompatibilities between the monotonicity axioms and
extended/proportional justified representation for the voting rules
that we analyze and provide formal proofs of incompatibility between
some of these axioms and perfect representation (another axiom
proposed by \citeauthor{pjr-aaai}~\shortcite{pjr-aaai}). At the end of
this paper we review briefly some previous works that study 
monotonicity axioms in approval-based multi-winner elections, draw
some conclusions and outline some lines of continuation of this work.

\section{Preliminaries}

We consider elections in which a fixed number $k$ of candidates or
alternatives must be chosen from a set of candidates $C$. We assume
that $|C| \geq k \geq 1$. The set of voters is represented as $N= \{1,
\ldots, n\}$.  Each voter $i$ that participates in the election casts
a ballot $A_i$ that consists of the subset of the candidates that the
voter approves of (that is, $A_i \subseteq C$). We refer to the
ballots cast by the voters that participate in the election as the
ballot profile $\mathcal{A}= (A_1, \ldots, A_n)$.  An approval-based
multi-winner election $\mathcal{E}$ is therefore represented by
$\mathcal{E}= (N, C, \mathcal{A}, k)$. The set of voters $N$ and the
set of candidates $C$ will be omitted when they are clear from the
context.

Given a voting rule $R$, for each election $\mathcal{E}= (\mathcal{A},
k)$, we say that $R(\mathcal{E})$ is the output of the voting rule $R$
for such election. Ties may happen in the voting rules that we are
going to consider. To take this into account, given an election
$\mathcal{E}$ and a voting rule $R$ we say that the value of
$R(\mathcal{E})$ is the set of size at least one composed of all the
possible sets of winners outputted by rule $R$ for election
$\mathcal{E}$. We say that a candidates subset $W$ of size $k$ is a
set of winners for election $\mathcal{E}$ and rule $R$ if $W$ belongs
to $R(\mathcal{E})$. We stress that our results are to a large extent
independent of how ties are broken.

Given an election $\mathcal{E}= (N, C, \mathcal{A}, k)$ and a
non-empty candidates subset $G$ of $C$, we define $\mathcal{E}_{\Delta
  G}$, as the election obtained by adding to election $\mathcal{E}$
one voter that approves of only the candidates in $G$. That
is, $\mathcal{E}_{\Delta G}= (N_{\Delta G}= \{1, \ldots, n, n+1\},
C, \mathcal{A}_{\Delta G}=(A_1, \ldots, A_n, G), k)$. Given a
non-empty candidates subset $G$ and a voter $i \in N$ such that she
does not approve of any of the candidates in $G$ we define
$\mathcal{E}_{i+G}$, as the election obtained if voter $i$ decides to
approve of all the candidates in $G$ in addition to the candidates in
$A_i$. That is, $\mathcal{E}_{i+G}= (N, C, \mathcal{A}_{i+G}= (A_1,
\ldots, A_{i-1}, A_i \cup G, A_{i+1}, \ldots, A_n), k)$.

We recall now the notions of justified representation and extended
justified representation due to
\citeauthor{aziz:scw}~\shortcite{aziz:scw}, and of proportional
justified representation due to
\citeauthor{pjr-aaai}~\shortcite{pjr-aaai}.

\begin{definition}
\label{def:jr}  
Consider an election $\mathcal{E} = (N, C, \mathcal{A}, k)$.  Given a
positive integer $\ell\in \{1, \ldots, k\}$, we say that a set of
voters $N^*\subseteq N$ is {\em $\ell$-cohesive} if $|N^*| \geq \ell
\frac{n}{k}$ and $|\bigcap_{i \in N^*} A_i| \geq \ell$. We say that a
set of candidates $W$, $|W| = k$, provides {\em justified
  representation} (JR) for $\mathcal{E}$ if for every $1$-cohesive set
of voters $N^* \subseteq N$ it holds that there exists a voter $i$ in $N^*$
such that $A_i \cap W \ne \emptyset$. We say that a set of
candidates $W$, $|W| = k$, provides {\em extended justified
  representation} (EJR) (respectively, {\em proportional justified
  representation} (PJR)) for $\mathcal{E}$ if for every $\ell\in\{1,
\ldots, k\}$ and every $\ell$-cohesive set of voters $N^* \subseteq N$
it holds that there exists a voter $i$ in $N^*$ such that $|A_i \cap W| \geq
\ell$ (respectively, $|W \cap (\bigcup_{i \in N^*} A_i)| \ge
\ell$). We say that an approval-based voting rule satisfies {\em
  justified representation} (JR), {\em extended justified representation}
  (EJR), or {\em proportional justified representation} (PJR) if for every
election $\mathcal{E} = (N, C, \mathcal{A}, k)$ it outputs a committee
that provides JR, EJR, or PJR, respectively, for $\mathcal{E}$.
\end{definition}

\citeauthor{aziz:scw}~\shortcite{aziz:scw}, and
\citeauthor{pjr-aaai}~\shortcite{pjr-aaai} prove that EJR implies PJR
and that PJR implies JR, both for rules and for committees.

Below we introduce the voting rules that we are going to consider in
this study. First of all, we present the following voting rules,
surveyed by \citeauthor{kilgour10}~\shortcite{kilgour10}.

{\bf Approval Voting (AV)}. Under AV, the winners are the $k$
candidates that receive the largest number of votes. Formally, for
each approval-based multi-winner election $(\mathcal{A}, k)$, the
approval score of a candidate $c$ is $|\{i: c \in A_i\}|$. The $k$
candidates with the highest approval scores are chosen.

{\bf Satisfaction Approval Voting (SAV)}. A voter's {\it satisfaction
  score} is the fraction of her approved candidates that are
elected. SAV maximizes the sum of the voters' satisfaction
scores. Formally, for each approval-based multi-winner election
$\mathcal{E}=(\mathcal{A}, k)$:

\begin{equation}
\textrm{SAV}(\mathcal{E}) = \underset{W \subseteq C:
  |W| = k}{\textrm{argmax}} \ 
\sum_{i \in N} \frac{|A_i \cap W|}{|A_i|}.
\end{equation}

{\bf Minimax Approval Voting (MAV)}. MAV selects the set of candidates
$W$ that minimizes the maximum {\it Hamming
  distance}~\cite{hamming1950error} between $W$ and
the voters' ballots. Let
$d(A,B) = |A \setminus B| + |B \setminus A|$, for each pair of
candidates subsets $A$ and $B$. Then,
for each approval-based multi-winner election $\mathcal{E} =(\mathcal{A}, k)$:

\begin{equation}
\textrm{MAV}(\mathcal{E})= \underset{W \subseteq C: |W| = k}
{\textrm{argmin}} \Big( \max_{i \in N} d(W, A_i) \Big).
\end{equation}

Since we are interested in the compatibility between representation
axioms and monotonicity axioms we are also going to study several
rules that satisfy some of the above mentioned representation axioms.

{\bf Chamberlin and Courant rule and Monroe rule}. The voting rules
proposed by \citeauthor{chamberlin}~\shortcite{chamberlin} and
\citeauthor{monroe}~\shortcite{monroe} select sets of winners that
minimize the misrepresentation of the voters (the number of voters
represented by a candidate that they do not approve of). The
difference between the rule of Chamberlin and Courant (CC) and the
rule of Monroe is that in CC each candidate may represent an arbitrary
number of voters while in the Monroe rule each candidate must
represent at least $\lfloor \frac{n}{k} \rfloor$ and at most $\lceil
\frac{n}{k} \rceil$ voters. For each approval-based multi-winner
election $\mathcal{E}=(\mathcal{A}, k)$:

\begin{equation}
\textrm{CC}(\mathcal{E}) = \displaystyle 
\underset{W \subseteq C: |W| = k}
{\textrm{argmin}}\ |\{i: A_i \cap W = \emptyset\}|.
\end{equation}

Given an election $\mathcal{E}=(\mathcal{A}, k)$ and a candidates subset
$W$ of size $k$ let $M_{N,W}$ be the set of all mappings $\pi:
N \rightarrow W$ such that for each candidate $c$ in $W$ it holds that
$\lfloor \frac{n}{k} \rfloor \leq
|\{i: \pi(i)=c\}| \leq \lceil \frac{n}{k} \rceil$. Then,

\begin{equation}
\textrm{Monroe}(\mathcal{E}) = \displaystyle 
\underset{W \subseteq C: |W| = k}
{\textrm{argmin}} \min_{\pi \in M_{N,W}} |\{i: \pi(i) \notin A_i \}|.
\end{equation}

{\bf Proportional Approval Voting (PAV)}
and {\bf Sequential Proportional Approval Voting (SeqPAV)} were
proposed by the Danish mathematician
\citeauthor{thiele:pav-rav}~\shortcite{thiele:pav-rav} in the late
19th century. Given an election $\mathcal{E}=(\mathcal{A},k)$ and a
candidates subset $W$ of size $k$, the PAV-score of a voter $i$ is 0
if such voter does not approve of any of the candidates in $W$ and
$\sum_{j=1}^{|A_i \cap W|} \frac{1}{j}$ if the voter approves of some
of the candidates in $W$. PAV selects the sets of winners that
maximize the sum of the PAV-scores of the voters.

\begin{equation}
\textrm{PAV}(\mathcal{E}) = \displaystyle 
\underset{W \subseteq C: |W| = k}
{\textrm{argmax}} \sum_{i: A_i \cap W \ne \emptyset} 
\sum_{j=1}^{|A_i \cap W|} \frac{1}{j}. 
\end{equation}

Under the SeqPAV, the set of winners is computed with an iterative
algorithm in which at each iteration the candidate with highest SeqPAV
score is added to the set of winners. The SeqPAV score of a candidate
$c$ at iteration $j$ is computed as follows:

\begin{equation}
\textrm{SeqPAV}(c) = \sum_{i:c \in A_i} 
 \frac{1}{1 + |A_i \cap W_{j-1}|}.
\label{eq:rav}
\end{equation}

Here $W_{j-1}$ is the set of the first $j-1$ candidates added by
SeqPAV to the set of winners.

{\bf Phragm\'en rules} Phragm\'en rules were proposed by the Swedish
mathematician~\citeauthor{phragmen:p1}~\shortcite{phragmen:p1,phragmen:p2,phragmen:pOpt,phragmen:p4}
in the late 19th century. We refer to the survey by
\citeauthor{2016arXiv161108826J}~\shortcite{2016arXiv161108826J} for
an extensive discussion of Phragm\'en rules.

Phragm\'en voting rules are based on the concept of {\it load}. Each
candidate in the set of winners incurs in one unit of load, that
should be distributed among the voters that approve of such
candidate. The goal is to choose the set of winners such that the
total load is distributed as evenly as possible between the voters.

Formally, given an election $\mathcal{E} = (\mathcal{A}, k)$ and a
candidates subset $W \subseteq C$, $|W|=k$, a load distribution is
a two dimensional array ${\bf x}= (x_{i,c})_{i \in N, c \in W}$,
that satisfies the following three conditions:

\begin{eqnarray}
0 \leq x_{i,c} \leq 1 & & \textrm{for all } i \in N \textrm{ and } c \in W, 
\label{eq:load1} \\
x_{i,c}= 0 & & \textrm{if } c \notin A_i, \label{eq:load2} \textrm{and} \\
\sum_{i \in N} x_{i,c} = 1 & & \textrm{for all } c \in W.
\label{eq:load4} 
\end{eqnarray}

Given a load distribution ${\bf x}$, the load of each voter $i$ is
defined as $x_i= \sum_{c \in W} x_{i,c}$. Then, given an election
$\mathcal{E}$, the rule max-Phragm\'en outputs the set of winners $W$
that minimizes the maximum voter load.

Under seq-Phragm\'en the load of each voter changes (increases) at
each iteration as candidates are added to the set of winners. For
$j=0, \ldots, k$, let $x_i^{(j)}$ be the load of voter $i$ after $j$
iterations of seq-Phragm\'en. The initial load $x_i^{(0)}$ of each
voter $i$ is set to 0.

At each iteration $j+1$ the load $s_c^{(j+1)}$ associated to each
  candidate $c$ is computed as:

\begin{equation}
s_c^{(j+1)}= \frac{1 + \sum_{i: c \in A_i} x_i^{(j)}}{|\{i: c \in A_i\}|}.
\label{eq:seq-phr}
\end{equation}

The underlying idea of this expression is to distribute equally
between all the voters that approve of candidate $c$ the unit of load
corresponding to such candidate plus the load that each of such voters
had after the first $j$ iterations. Then, at each iteration the
candidate $w$ with the lowest load is added to the set of winners and
the loads of the voters are updated as follows: for each voter $i$
that approves of candidate $w$, we have $x_i^{(j+1)}= s_w^{(j+1)}$,
while the load of each voter $h$ that does not approve of candidate
$w$ does not change: $x_h^{(j+1)}= x_h^{(j)}$.

\section{Support monotonicity}

Some previous work (see Section~\ref{sec:relw}) make use of the
following idea of support monotonicity: if a candidate that was
already in the set of winners is added to the ballot of some voter
(without changing anything else in the election), then such candidate
must still belong to the set of winners. We will refer to this axiom
as {\it candidate monotonicity}.

\begin{definition}
We say that a rule $R$ satisfies {\em candidate monotonicity} if for
each election $\mathcal{E} = (N, C, \mathcal{A}, k)$, for each
candidate $c \in C$, and for each voter $i$ that does not approve of
$c$, the following conditions hold: (i) if $c$ belongs to some winning
committee in $R(\mathcal{E})$, then $c$ must also belong to some
winning committee in $R(\mathcal{E}_{i+\{c\}})$; and (ii) if $c$
belongs to all winning committees in $R(\mathcal{E})$, then $c$ must
also belong to all winning committees in $R(\mathcal{E}_{i+\{c\}})$.
\end{definition}

Candidate monotonicity can be seen as the equivalent of the axiom with
the same name proposed by
\citeauthor{elkind:scw17}~\shortcite{elkind:scw17} for ranked
ballots\footnote{\citeauthor{elkind:scw17}~\shortcite{elkind:scw17}
  also proposed another axiom called non-crossing monotonicity that
  will not be considered in this paper.}. They require that, if a
winning candidate $c$ is moved forward in some vote, then $c$ must still
belong to some winning committee.

\citeauthor{elkind:scw17}~\shortcite{elkind:scw17} justified
this axiom with the following idea:

\begin{quote}
  
``If $c$ belongs to a winning committee $W$ then, generally speaking,
  we cannot expect $W$ to remain winning when $c$ is moved forward in
  some vote, as this shift may hurt other members of $W$.''

\end{quote}

In this paper we propose to extend the notion of candidate
monotonicity for approval-based multi-winner voting rules in several
directions. First of all, we study what happens when a subset $G$ of
the candidates that was already in the set of winners $W$ is added to
the ballot of some voter. Following the idea of
\citeauthor{elkind:scw17}~\shortcite{elkind:scw17} that we have quoted
before, we believe that we cannot expect $W$ to remain winning, but we
can expect that all the candidates in $G$ (strong version) or, at
least, some of the candidates in $G$ (weak version) remain winning.

Secondly, we consider monotonicity when a new voter enters
the election and approves of a subset of the candidates that was
already in the set of winners. Again, we define a strong and a weak
version of this axiom. A similar idea of monotonicity when new voters
enter the election has been proposed by
\citeauthor{woodall:prop}~\shortcite{woodall:prop} for ranked ballots.

\begin{definition}
We say that a rule $R$ satisfies {\em strong support monotonicity with
  population increase} (respectively, {\em weak support monotonicity
  with population increase}) if for each election $\mathcal{E} = (N,
C, \mathcal{A}, k)$, and for each non-empty subset $G$ of $C$, such
that $|G| \leq k$, the following conditions hold: (i) if $G \subseteq
W$ for some $W \in R(\mathcal{E})$, then $G \subseteq W'$ for some $W'
\in R(\mathcal{E}_{\Delta G})$ (respectively, $G \cap W' \ne
\emptyset$ for some $W' \in R(\mathcal{E}_{\Delta G})$); and (ii) if
$G \subseteq W$ for all $W \in R(\mathcal{E})$, then $G \subseteq W'$
for all $W' \in R(\mathcal{E}_{\Delta G})$ (respectively, $G \cap W'
\ne \emptyset$ for all $W' \in R(\mathcal{E}_{\Delta G})$).

We say that a rule $R$ satisfies {\em strong support monotonicity
  without population increase} (respectively, {\em weak support
  monotonicity without population increase}) if for each election
$\mathcal{E} = (N, C, \mathcal{A}, k)$, for each non-empty subset $G$
of $C$, such that $|G| \leq k$, and for each voter $i$ such that $A_i
\cap G = \emptyset$, the following conditions hold: (i) if $G
\subseteq W$ for some $W \in R(\mathcal{E})$, then $G \subseteq W'$
for some $W' \in R(\mathcal{E}_{i+G})$ (respectively, $G \cap W' \ne
\emptyset$ for some $W' \in R(\mathcal{E}_{i+G})$); and (ii) if $G
\subseteq W$ for all $W \in R(\mathcal{E})$, then $G \subseteq W'$ for
all $W' \in R(\mathcal{E}_{i+G})$ (respectively, $G \cap W' \ne
\emptyset$ for all $W' \in R(\mathcal{E}_{i+G})$).
\end{definition}

We believe that it is important to know what happens when the support
of several of the candidates in the set of winners is incremented
simultaneously. Moreover, our results show that for each of the rules
that we consider that satisfies any of the support monotonicity axioms
(with or without population increase) for $|G|=1$, such rule also
satisfies the corresponding weak support monotonicity axiom (for all
values of $|G|$), which is slightly stronger, and therefore provides
more information about the behaviour of the rule. Because of this, we
do not study candidate monotonicity in this paper. We note, however,
that we have been able to build (weird) rules that satisfy support
monotonicity with or without population increase for $|G|=1$ but fail
the corresponding weak axiom (examples can be found in
appendix~\ref{ap:examples}).

We now discuss briefly the relevance of these axioms for the three
types of scenarios considered in the Introduction. First of all we
note that it is a general property of elections to desire to select
winners that receive a high support, and therefore we believe that our
weak axioms are generally desirable.

In the case of excellence, we believe that the strong axioms are
highly preferable to the weak ones. Since we are looking for the best
candidates, adding support to a subset of the candidates that were
already considered to be among the best should make all of them stay
in the set of winners.

In contrast, in the case of diversity we believe that satisfying the
strong axioms is not important for the rule used in the election. We
recall that the goal of the election is that every voter has one of
her preferred candidates in the set of winners. If the support of a
subset $G$ of the winners was increased and at least one of them
remained in the set of winners, then the voters that approve of the
candidates in $G$ would be satisfied. Removing some of the candidates
in $G$ from the set of winners may allow to add other candidates
approved by other voters that did not have previously any of their
approved candidates in the set of winners. Therefore, for an election
of the diversity type, we believe that it would be enough if the rule
satisfies the weak axioms.

\citeauthor{pjr-aaai}~\shortcite{pjr-aaai} distinguish two types of
proportional representation. In the first type of proportional
representation the aim is that each voter is represented by a
candidate that she approves of and that each candidate represents the
same number of voters. The typical example of this type of scenario
are parliamentary elections. As in the case of diversity, in this type
of scenario we believe that it is enough to satisfy the weak
axioms. Regarding the second type of proportional representation
considered by \citeauthor{pjr-aaai}~\shortcite{pjr-aaai}, the goal is
that, for each $\ell$-cohesive group of voters (see
Definition~\ref{def:jr}), as most voters of the group as possible
approve of at least $\ell$ of the candidates in the set of
winners. \citeauthor{pjr-aaai}~\shortcite{pjr-aaai} present as an
example of this type of elections the selection of researchers invited
to give a seminar in an academic department. We believe that the
situation in this case is less clear. It seems that the weak axioms
are not enough for this situation because a voter may not be satisfied
with having only one of her preferred candidates in the set of
winners. However, the strong axioms are maybe too strong if a voter
that belongs to an $\ell$-cohesive group of voters decides to approve
of a subset of the set of winners $G$ of size greater than $\ell$.

From now on, we will refer to support monotonicity with population
increase as SMWPI and to support monotonicity without population
increase as SMWOPI. Table~\ref{tab:comparisonSummary} summarizes the
results we have obtained in this paper. With respect to the support
monotonicity axioms (columns entitled ``SMWPI'' and ``SMWOPI'') we use
the keys ``Str.'' when the rule satisfies the strong version of the
axiom, ``Wk.'' when the rule satisfies the weak version of the axiom
and ``No'' when the rule does not satisfy any of them. The column
entitled ``C. M.'' contains the results related to committee
monotonicity, which is discussed in Section~\ref{sec:comMon}.

For completeness, we also include previous results related to the
computational complexity of the rules and the representation axioms
that they satisfy, including pointers to the appropriate
references. The column entitled ``JR/PJR/EJR'' shows for each rule the
strongest of these axioms satisfied by the rule. The next column says
which rules satisfy the perfect representation axiom (PR), that will
be discussed in Section~\ref{sec:comp}.

\ctable[
pos= htb,
star,
caption= {Properties of approval-based 
multi-winner voting 
rules\label{tab:comparisonSummary}},
botcap
]{|l|l|l|l|l|l|l|} 
{\tnote[a]{Results taken from~\cite{aziz2014computational} 
and~\cite{skowron2016finding}.}
\tnote[b]{Results taken from~\cite{procaccia:complex}.}
\tnote[c]{Results taken from~\cite{brill:phragmen}.}
\tnote[d]{Results taken from~\cite{aziz:scw}.}
\tnote[e]{Monroe satisfies PJR if $k$ divides $n$~\cite{pjr-aaai}.} 
\tnote[f]{max-Phragm\'en satisfies PJR when combined with certain
    tie-breaking rule~\cite{brill:phragmen}.}
\tnote[g]{CC satisfies PR if ties are broken always in favour of
the candidates subsets that provide PR.}
\tnote[h]{Results taken from~\cite{pjr-aaai}.}
\tnote[i]{Results taken from~\cite{2016arXiv161108826J},
\cite{mora2015eleccions}, and~\cite{phragmen:pOpt}.}
\tnote[j]{Results taken from~\cite{thiele:pav-rav}.}
\tnote[k]{Results taken from~\cite{legrand2006approval}.}
}
{\hline
{\bf Rule} & {\bf Complexity} & {\bf JR/PJR/EJR} & {\bf PR}
& {\bf SMWPI} & {\bf SMWOPI} & {\bf C. M.}\\ \hline
AV & P\tmark[\ a] & No\tmark[\ d] & No\tmark[Ex.~\ref{ex:pr-av}] & Str.\tmark[Thm.~\ref{th:one}] & Str.\tmark[Thm.~\ref{th:three}] & Yes\\ \hline
SAV & P\tmark[\ a] & No\tmark[\ d] & No\tmark[Ex.~\ref{ex:pr-av}] & Str.\tmark[Thm.~\ref{th:one}] & Str.\tmark[Thm.~\ref{th:three}] & Yes\\ \hline
MAV & NP-hard\tmark[k] & No\tmark[d] & No\tmark[Ex.~\ref{ex:pr-mav}] & Str.\tmark[Thm.~\ref{th:mav}] & Wk.\tmark[Thm.~\ref{th:mav}] & No\tmark[Ex.~\ref{ex:cm-mav}]\\ \hline
CC & NP-comp.\tmark[b] & JR\tmark[\ d] & Yes\tmark[g, Ex.~\ref{ex:pr-cc}] & 
Str.\tmark[Thm.~\ref{th:one}] & Wk.\tmark[Thm.~\ref{th:two}] & No\tmark[Ex.~\ref{ex:cm-monroecc}]\\ \hline
Monroe & NP-comp.\tmark[b] & JR\tmark[\ d,e] & Yes\tmark[\ h] &
No\tmark[Ex.~\ref{ex:monWSMWPI}] & Wk.\tmark[Thm.~\ref{th:four}] & No\tmark[Ex.~\ref{ex:cm-monroecc}]\\ \hline
PAV & NP-comp.\tmark[a] & EJR\tmark[\ d] & No\tmark[\ h] & Str.\tmark[Thm.~\ref{th:one}] & Wk.\tmark[Thm.~\ref{th:two}] & No\tmark[\ j]\\ \hline
SeqPAV & P\tmark[a] & No\tmark[d,h] & No\tmark[Ex.~\ref{ex:pr-seqpav}] & Wk.\tmark[Thm.~\ref{th:seqpav}] & Wk.\tmark[Thm.~\ref{th:seqpav}] & Yes\\ \hline
max-Phragm\'en & NP-comp.\tmark[c] & PJR\tmark[\ c,f] & Yes\tmark[\ c]
& Wk.\tmark[i, Thm.~\ref{th:six}] & Wk.\tmark[i, Thm.~\ref{th:six}] & No\tmark[\ i]\\ \hline
seq-Phragm\'en & P\tmark[\ c] & PJR\tmark[\ c] & No\tmark[\ c] & 
Wk.\tmark[i] & Wk.\tmark[i] & Yes \\ \hline
}

An important type of rules in approval-based multi-winner elections are
{\em approval-based multi-winner counting rules}, which, as discussed
by \citeauthor{2017arXiv170402453L}~\shortcite{2017arXiv170402453L,lackner2018approval},
can be seen as analogous to the class of committee scoring rules
introduced by \citeauthor{elkind:scw17}~\shortcite{elkind:scw17} for
ranked-based multi-winner elections.

\begin{definition}
A counting function $f: \{1, \ldots, k\} \times \{1, \ldots,
|C|\} \rightarrow \mathbb{R}$ is a function that satisfies that
$f(x,y) \geq f(x',y)$ whenever $x > x'$. Intuitively, a counting
function $f$ defines the score $f(x,y)$ that a certain counting rule
$r_f$ assigns to a voter $i$ that approves of $x$ candidates in the
set of winners $W$ and $y$ candidates in total. Given a counting
function $f$, and an election $\mathcal{E}= (\mathcal{A}, k)$, the
total score of a candidates subset $W$ for counting function $f$ is

\begin{displaymath}   
s_f(W, \mathcal{E})= \sum_{i \in N} f(|A_i \cap W|, |A_i|), 
\end{displaymath}   

and the counting rule $r_f$ associated to counting function $f$ is
defined as follows:

\begin{displaymath}   
r_f(\mathcal{E})= \underset{W \subseteq C: |W| = k}{\textrm{argmax}} 
s_f(W, \mathcal{E}).
\end{displaymath}   

\end{definition} 

As discussed by
\citeauthor{2017arXiv170402453L}~\shortcite{2017arXiv170402453L,lackner2018approval}
several of the voting rules that we have presented in the previous
section are counting rules. In particular, we have
$f_{\textrm{AV}}(x,y)= x$ for AV, $f_{\textrm{SAV}}(x,y)= \frac{x}{y}$
for SAV, $f_{\textrm{CC}}(x,y)= 1$ if $x > 0$ and
$f_{\textrm{CC}}(0,y)= 0$ for CC, and $f_{\textrm{PAV}}(x,y)=
\sum_{j=1}^{x} \frac{1}{j}$ if $x>0$ and $f_{\textrm{PAV}}(0,y)= 0$
for PAV.

For counting rules we have the following results with respect to
support monotonicity.

\begin{theorem}
  Every counting rule satisfies strong SMWPI.
  \label{th:one}
\end{theorem}

\begin{proof}[Proof]
Consider an election $\mathcal{E} =(\mathcal{A}, k)$, a counting
function $f$ and its associated rule $r_f$, a set of winners $W$
outputted by $r_f$ for election $\mathcal{E}$ and a non-empty subset
$G$ of $W$. We are going to prove that $W$ also belongs to
$r_f(\mathcal{E}_{\Delta G})$. The theorem follows from that
immediately.

Consider any other candidates subset $W'$ of size $k$. We simply have
to observe that the total score of $W$ for election
$\mathcal{E}_{\Delta G}$ under rule $r_f$ is $\sum_{i \in N}
f(|A_i \cap W|, |A_i|) + f(|G \cap W|,|G|)$, that $\sum_{i \in N}
f(|A_i \cap W|, |A_i|) \geq \sum_{i \in N} f(|A_i \cap W'|, |A_i|)$
(because $W$ is a set of winners for rule $r_f$ and election
$\mathcal{E}$), and that $f(|G \cap W|,|G|) = f(|G|,|G|) \geq f(|G \cap
W'|,|G|)$ (by the definition of counting function).
\end{proof}

We can also prove weak SMWOPI by introducing a slight restriction to
the counting functions that is satisfied by all the counting rules
that we consider in this paper.

\begin{theorem}
\label{th:two}  
Consider a counting function $f$. If $f$ satisfies that $f(x,y) \geq
f(x,y')$ whenever $y \leq y'$, and that for each positive integer $z$
it holds that $f(x+z,y+z) \geq f(x,y)$, then its associated rule
$r_f$ satisfies weak SMWOPI. 
\end{theorem}

\begin{proof}[Proof]
Consider an election $\mathcal{E} =(\mathcal{A}, k)$, a counting
function $f$ and its associated rule $r_f$, a set of winners $W$
outputted by $r_f$ for election $\mathcal{E}$, a non-empty subset $G$
of $W$, and a voter $i$ such that she does not approve of the
candidates in $G$.

We observe first that because $A_i$ and $G$ are disjoint, for each
candidates subset $W'$ it holds that $f(|(A_{i} \cup G) \cap W'|,
|A_{i} \cup G|)= f(|A_{i} \cap W'| + |G \cap W'|, |A_{i}| + |G|)$ and
that $s_f(W',\mathcal{E}_{i+G})-s_f(W',\mathcal{E})= f(|A_{i} \cap W'|
+ |G \cap W'|, |A_{i}| + |G|) - f(|A_{i} \cap W'|, |A_{i}|)$.

Suppose that $f$ satisfies that $f(x,y) \geq f(x,y')$ whenever
$y \leq y'$, and that for each positive integer $z$ it holds that
$f(x+z,y+z) \geq f(x,y)$, and consider any candidates subset $W'$ of
size $k$ such that $W' \cap G = \emptyset$. Then,

\begin{displaymath}
\begin{array}{l}  
  s_f(W,\mathcal{E}_{i+G})-s_f(W',\mathcal{E}_{i+G})= \\
\ \ (s_f(W,\mathcal{E})+ f(|A_{i} \cap W|
+ |G \cap W|, |A_{i}| + |G|) \\
\ \ - f(|A_{i} \cap W|, |A_{i}|)) \\ 
\ \ - (s_f(W',\mathcal{E})+ f(|A_{i} \cap W'|
+ |G \cap W'|, |A_{i}| + |G|) \\
\ \ - f(|A_{i} \cap W'|, |A_{i}|)) \\  
\ \ = s_f(W,\mathcal{E})-s_f(W',\mathcal{E}) \\
\ \ + f(|(A_{i} \cap W| + |G|,
|A_{i}| + |G|) - f(|(A_{i} \cap W|, |A_{i}|) \\
\ \ - (f(|A_{i} \cap W'|,
|A_{i}| + |G|) - f(|A_{i} \cap W'|, |A_{i}|)) \geq 0.
\end{array}
\end{displaymath}

This proves part (i) of the definition of weak SMWOPI. The proof of
part (ii) of the definition of weak SMWOPI follows from the fact that
if $G \subseteq W$ for all $W \in r_f(\mathcal{E})$, then
$s_f(W,\mathcal{E})-s_f(W',\mathcal{E}) > 0$, and therefore, the
inequality in the equation above is strict.

\end{proof}

However, of the counting rules that we consider in this paper only AV
and SAV satisfy strong SMWOPI.

\begin{theorem}
\label{th:three}  
AV and SAV satisfy strong SMWOPI.
\end{theorem}

\begin{proof}[Proof]
The counting functions of AV and SAV hold that $f(x,y)= xf(1,y)$. This
makes it possible to assign each candidate $c$ a score
$s_f(c,\mathcal{E})= \sum_{i: c \in A_i} f(1,|A_i|)$ irrespective of
which other candidates are in the set of winners $W$ so that
$s_f(W,\mathcal{E})= \sum_{c \in W} s_f(c,\mathcal{E})$. Therefore,
the winners in AV and SAV are the $k$ candidates with the highest
candidate score. It is now enough to observe that each candidate that
belongs to $G$ increases her score in election $\mathcal{E}_{i+G}$ in
$f(1,|A_i|+|G|)$ with respect to her score in election $\mathcal{E}$,
and that the scores of the candidates that are not in $G$ do not
increase.
\end{proof}

The following examples prove that PAV and CC fail strong SMWOPI.

\begin{example}
Let $k=4$ and $C=\{c_1, \ldots, c_7\}$. $131$ voters cast the
following ballots: for $i,j= 1$ to $3$, $3$ voters approve of
$\{c_i,c_{j+4}\}$, $100$ voters approve of $\{c_4\}$, $1$ voter
approves of $\{c_1, c_2\}$, $1$ voter approves of $\{c_1, c_2, c_3\}$,
and $2$ voters approve of $\{c_5,c_6\}$.  For this election PAV
outputs one set of winners: $\{c_1, c_2, c_3, c_4\}$, with a PAV score
of $391/3$. However, if the voter that approves of $\{c_1, c_2\}$
decides to approve of $\{c_1, c_2, c_3, c_4\}$, then PAV outputs only
$\{c_4, c_5, c_6, c_7\}$, with a PAV score of $131$. Intuitively, this
example works as follows. First, the $100$ voters that approve of
$\{c_4\}$ force that $c_4$ has to be in the set of winners. Second,
the first $27$ votes force that either $\{c_1, c_2, c_3\}$ or $\{c_5,
c_6, c_7\}$ are in the set of winners. The last $4$ votes break the
tie between $\{c_1, c_2, c_3, c_4\}$ and $\{c_4, c_5, c_6, c_7\}$ in
the two cases considered.
\end{example}

\begin{example}
Let $k=3$ and $C=\{a, b, c, d, e\}$. $13$ voters cast the following
ballots: $2$ voters approve of $\{a, d\}$, $2$ voters approve of $\{a,
e\}$, $2$ voters approve of $\{c, d\}$, $2$ voters approve of $\{c,
e\}$, $2$ voters approve of $\{b\}$, $2$ voters approve of $\{a\}$,
and $1$ voter approves of $\{d\}$. For this election CC outputs one
set of winners: $\{a, b, c\}$ (one voter misrepresented). Now, we
consider two consecutive increases of support of $\{b, c\}$, where, in
each increase one of the voters that approve of $\{a\}$ decides to
approve of $\{a, b, c\}$. Then, after the first increase of support of
$\{b, c\}$, CC outputs $\{a, b, c\}$ and $\{b, d, e\}$ (one voter
misrepresented), and after the second increase of support of $\{b,
c\}$ CC outputs only $\{b, d, e\}$ ($0$ voters
misrepresented). Observe that this example proves that CC fails strong
SMWOPI even when it is combined with any tie breaking rule, because if
the tie breaking rule selects $\{a, b, c\}$ after the first increase
of support, then strong SMWOPI is violated in the second increase of
support, and if the tie breaking rule selects $\{b, d, e\}$ after the
first increase of support, then strong SMWOPI is violated in the first
increase of support.
\end{example}

Let us now turn to analyze the remaining voting rules. 

\begin{theorem}
\label{th:mav}MAV satisfies strong SMWPI and weak SMWOPI.  
\end{theorem}

\begin{proof}[Proof]
Consider first an election $\mathcal{E} =(\mathcal{A}, k)$, a set of
winners $W$ outputted by $\textrm{MAV}$ for election $\mathcal{E}$,
and a non-empty subset $G$ of $W$. Since $G \subseteq W$, we have
$d(W,G)= |W \setminus G| + |G \setminus W|= k - |G|$. For each
candidate set $W'$ of size $k$ such that a candidate $c$ exists that
belongs to $G$ but not to $W'$, we have $d(W',G)= |W' \setminus G| +
|G \setminus W'| \geq (k - |G| + 1) + 1= k - |G| + 2$. Therefore,
$\max \{\max_{i \in N} d(W, A_i), d(W,G)\}$ has to be less than or
equal to\\ $\max \{\max_{i \in N} d(W', A_i), d(W',G)\}$. This proves
that MAV satisfies strong population monotonicity with population
increase.

Consider now an election $\mathcal{E} =(\mathcal{A}, k)$, a set of
winners $W$ outputted by $\textrm{MAV}$ for election $\mathcal{E}$, a
non-empty subset $G$ of $W$, and a voter $i$ that does not aprove of
any of the candidates in $G$. We observe first that $(A_i \cup
G) \setminus W= A_i \setminus W$, and that $|W \setminus (A_i \cup
G)|= |W \setminus A_i| - |G|$, and therefore, $d(W, A_i \cup G)= d(W,
A_i) - |G|$. For each candidates set $W'$ of size $k$ such that
$W' \cap G= \emptyset$, we have $d(W',A_i \cup G)= |W' \setminus
(A_i \cup G)| + |(A_i \cup G) \setminus W'| = |W' \setminus A_i| +
|A_i \setminus W'| + |G|= d(W',A_i) + |G|$. It follows immediately
that for each candidates set $W'$ of size $k$ such that $W' \cap
G= \emptyset$, the maximum Hamming distance between $W'$ and the
voters in election $\mathcal{E}_{i+G}$ does not decrease with respect
to the maximum Hamming distance between $W'$ and the voters in
election $\mathcal{E}$, and therefore, that $W$ or another set of
candidates that includes some of the candidates in $G$ must be output
by MAV for election $\mathcal{E}_{i+G}$.
\end{proof}

However, the following example shows that MAV fails strong SMWOPI.

\begin{example}
Let $k=5$ and $C= \{c_1, \ldots, c_7\}$. $9$ voters cast the following
ballots: $1$ voter approves of $\{c_1\}$, $1$ voter approves of
$\{c_2\}$, $1$ voter approves of $\{c_3\}$, $1$ voter approves of
$\{c_5\}$, for $i,j= 0$ to $1$, $1$ voter approves of $\{c_1, c_{4+i},
c_{6+j}\}$, and $1$ voter approves of $\{c_1, c_6\}$. For this
election the only set of winners output by MAV is $\{c_1, c_2, c_3,
c_4, c_5\}$. The Hamming distance between $\{c_1, c_2, c_3, c_4,
c_5\}$ and the ballot profile of each voter is always less than or
equal to $5$ (in particular, the Hamming distance between $\{c_1, c_2,
c_3, c_4, c_5\}$ and $\{c_1, c_6\}$ is $5$). We show now that for any
other candidate subset of size $5$ there is a ballot profile with
distance $6$ to such candidate subset. First, for each $j= 1, 2, 3$,
and for each candidate subset $W$ of size $5$ such that $c_j$ does not
belong to $W$, the Hamming distance between $W$ and $\{c_j\}$ is
$6$. There exist $6$ candidates subsets of size $5$ that contain
$c_1$, $c_2$, and $c_3$ (one of them is $\{c_1, c_2, c_3, c_4,
c_5\}$). For $i,j= 0$ to $1$, the Hamming distance between $\{c_1,
c_2, c_3, c_{4+i}, c_{6+j}\}$ and $\{c_1, c_{5-i}, c_{7-j}\}$ is
$6$. The only remaining candidates subset is $\{c_1, c_2, c_3, c_6,
c_7\}$, that has a Hamming distance with $\{c_5\}$ of $6$. Observe
that the Hamming distance between $\{c_1, c_2, c_3, c_6, c_7\}$ and
all the other ballot profiles is always less than or equal to
$4$. Now, if the voter that approves of $\{c_5\}$ decides to approve
of $\{c_1, c_2, c_3, c_4, c_5\}$, then the Hamming distance between
the ballot profile of such voter and $\{c_1, c_2, c_3, c_6, c_7\}$
falls to $4$, and therefore, in that case MAV would output only
$\{c_1, c_2, c_3, c_6, c_7\}$.
\end{example}

\begin{theorem}
\label{th:four}  
The Monroe rule satisfies weak SMWOPI.
\end{theorem}

\begin{proof}[Proof]
Consider an election $\mathcal{E} =(\mathcal{A}, k)$, a set of winners
$W$ outputted by $\textrm{Monroe}$ for election $\mathcal{E}$, a
non-empty subset $G$ of $W$, and a voter $i$ that does not approve of
any of the candidates in $G$. Let $\pi_W$ be a mapping that minimizes
the misrepresentation of $W$ for election $\mathcal{E}$. Clearly the
misrepresentation of $W$ with mapping $\pi_W$ for election
$\mathcal{E}_{i+G}$ is the same as for election $\mathcal{E}$ if the
candidate $\pi_W(i)$ assigned by $\pi_W$ to voter $i$ does not belong
to $G$ and is equal to the misrepresentation of $W$ with mapping
$\pi_W$ for election $\mathcal{E}$ minus one if $\pi_W(i)$ belongs to
$G$. Furthermore, for each candidates set $V$ such that $V \cap G=
\emptyset$, and for each mapping $\pi_{V}$ of the voters in $N$ to the
candidates in $V$ it holds that the candidate $\pi_{V}(i)$ assigned by
$\pi_V$ to voter $i$ belongs to $A_i \cup G$ if and only if such
candidate belongs to $A_i$ and, therefore, the misrepresentation
values of $V$ with mapping $\pi_{V}$ are the same for election
$\mathcal{E}_{i+G}$ and for election $\mathcal{E}$. 
\end{proof}

Examples~\ref{ex:monWSMWPI} and~\ref{ex:monSSMWOPI} prove that Monroe
fails weak SMWPI and strong SMWOPI, respectively. As in the case of
CC, these examples prove that Monroe fails weak SMWPI and strong
SMWOPI even if combined with any tie breaking rule.

\begin{example}
\label{ex:monWSMWPI}
Let $k=4$ and $C= \{a, b, c, d, e, f, g, h\}$. $33$ voters cast the
following ballots: $5$ voters approve of $\{a, e\}$, $4$ voters
approve of $\{a, g\}$, $5$ voters approve of $\{b, e\}$, $4$ voters
approve of $\{b, h\}$, $5$ voters approve of $\{c, f\}$, $4$ voters
approve of $\{c, g\}$, $3$ voters approve of $\{d, f\}$, and $3$
voters approve of $\{d, h\}$. For this election Monroe outputs only
$\{e, f, g, h\}$ (misrepresentation $1$ due to one of the voters that
approve of $e$ being represented by $h$). We now consider two
consecutive voters that enter the election, such that each of the new
voters approves of $\{e\}$. Then, after the first new voter enters the
election, Monroe outputs $\{a, b, c, d\}$ and $\{e, f, g, h\}$
(misrepresentation $2$) and, after the second new voter enters the
election, Monroe outputs only $\{a, b, c, d\}$ (misrepresentation $2$:
the new voters would be represented by candidate $d$).
\end{example}

\begin{example}
\label{ex:monSSMWOPI}
Let $k=3$ and $C= \{a, b, c, d, e\}$. $18$ voters cast the following
ballots: $2$ voters approve of $\{a\}$, $2$ voters approve of $\{a,
d\}$, $2$ voters approve of $\{a, e\}$, $4$ voters approve of $\{b\}$,
$1$ voter approves of $\{b, e\}$, $4$ voters approve of $\{c, d\}$,
and $3$ voters approve of $\{c, e\}$. For this election Monroe outputs
only $\{a, b, c\}$ (misrepresentation $1$ due to one of the voters
that approve of $c$ being represented by candidate $b$). Now, we
consider two consecutive increases of support of $\{b, c\}$, where, in
each increase one of the voters that approve of $\{a\}$ decides to
approve of $\{a, b, c\}$. Then, after the first increase of support of
$\{b, c\}$, Monroe outputs $\{a, b, c\}$ and $\{b, d, e\}$
(misrepresentation $1$), and after the second increase of support of
$\{b, c\}$ Monroe outputs only $\{b, d, e\}$ (misrepresentation $0$).
\end{example}

The results for SeqPAV and seq-Phragm\'en are studied together,
although for seq-Phragm\'en we need first an intermediate lemma.

\begin{lemma}
\label{lem:seqPh1}
Consider an election $\mathcal{E} =(\mathcal{A}, k)$, a set of winners
$W$ outputted by $\textrm{seq-Phragm\'en}$ for election $\mathcal{E}$,
a non-empty subset $G$ of $W$, and a voter $i$ that does not approve
of any of the candidates in $G$. Let $h$ be the first iteration in
which a candidate that belongs to $G$ is added to the set of winners
by $\textrm{seq-Phragm\'en}$, and let $c_h$ be such candidate. Then,
it holds that $s_{c_h}^{(h)} \geq \frac{1 + x_{i}^{(h-1)} + \sum_{r:
c_h \in A_{r}} x_{r}^{(h-1)}}{1 + |\{r: c_h \in A_{r}\}|}$.
\end{lemma}

\begin{proof}[Proof]
\citeauthor{brill:phragmen}~\shortcite{brill:phragmen} prove that for
each election $(\mathcal{A}, k)$, and for each $1 \leq j \leq k$, it
holds that $s^{(1)} \leq \ldots \leq s^{(k)}$, where $s^{(j)}$ is the
load $s_{c_j}^{(j)}$ of the candidate $c_j$ elected at iteration
$j$. Therefore, $s_{c_j}^{(j)} \geq x_r^{(j-1)}$ for each iteration
$j$ and each voter $r$. Thus, in the case of election $\mathcal{E}$
and iteration $h$ we have $s_{c_h}^{(h)}= \frac{1 + \sum_{r: c_h \in
    A_{r}} x_{r}^{(h-1)}}{|\{r: c_h \in A_{r}\}|} \geq \frac{1 +
  x_{i}^{(h-1)} + \sum_{r: c_h \in A_{r}} x_{r}^{(h-1)}}{1 + |\{r: c_h
  \in A_{r}\}|}$.
\end{proof}

The following theorem has already been proved by
\citeauthor{phragmen:pOpt}~\shortcite{phragmen:pOpt} and
\citeauthor{2016arXiv161108826J}~\shortcite{2016arXiv161108826J} for
seq-Phragm\'en in the case in which $|G|=1$. Our proof follows the
same ideas.

\begin{theorem}
  SeqPAV and seq-Phragm\'en satisfy weak SMWPI and weak SMWOPI.
  \label{th:seqpav}
\end{theorem}

\begin{proof}[Proof]
Consider an election $\mathcal{E} =(\mathcal{A}, k)$, a set of winners
$W$ outputted by $\textrm{SeqPAV}$ (respectively, by
$\textrm{seq-Phragm\'en}$) for election $\mathcal{E}$, a non-empty
subset $G$ of $W$, and a voter $i$ that does not aprove of any of the
candidates in $G$. Let $h$ be the first iteration in which a candidate
that belongs to $G$ is added to the set of winners by
$\textrm{SeqPAV}$ (respectively, by $\textrm{seq-Phragm\'en}$) and let
$c_h$ be such candidate. We observe first that while no candidate that
belongs to $G$ is added to the set of winners, the SeqPAV score (for
SeqPAV) and the load (for seq-Phragm\'en) of each candidate is the
same for elections $\mathcal{E}$, $\mathcal{E}_{\Delta G}$, and
$\mathcal{E}_{i+G}$. For each of $\mathcal{E}_{\Delta G}$ and
$\mathcal{E}_{i+G}$ there are therefore two possibilities: either a
candidate that belongs to $G$ is added to the set of winners by
$\textrm{SeqPAV}$ (respectively, by $\textrm{seq-Phragm\'en}$) in the
first $h-1$ iterations (in that case, the theorem holds) or the first
$h-1$ candidates added to the set of winners by $\textrm{SeqPAV}$
(respectively, by $\textrm{seq-Phragm\'en}$) for elections
$\mathcal{E}_{\Delta G}$ and $\mathcal{E}_{i+G}$ are the same (and
selected in the same order) as the first $h-1$ candidates added to the
set of winners by $\textrm{SeqPAV}$ (respectively, by
$\textrm{seq-Phragm\'en}$) for election $\mathcal{E}$. We now simply
observe that for $\mathcal{E}_{\Delta G}$ and $\mathcal{E}_{i+G}$ if
the first $h-1$ candidates added to the set of winners by
$\textrm{SeqPAV}$ (respectively, by $\textrm{seq-Phragm\'en}$) are the
same and in the same order as those added for election $\mathcal{E}$,
then at iteration $h$ the SeqPAV score of candidate $c_h$ increases
with respect to her SeqPAV score for election $\mathcal{E}$ and the
load of candidate $c_h$ (for seq-Phragm\'en) decreases with respect
her load for election $\mathcal{E}$ (in the case of election
$\mathcal{E}_{i+G}$ this follows from lemma~\ref{lem:seqPh1}), while
the SeqPAV score and the load of all the candidates that do not belong
to $G$ does not change. This proves that the candidate elected at
iteration $h$ both by SeqPAV and seq-Phragm\'en in elections
$\mathcal{E}_{\Delta G}$ and $\mathcal{E}_{i+G}$ must belong to $G$.
\end{proof}

The following example proves that SeqPAV and seq-Phragm\'en fail both
strong SMWPI and strong SMWOPI.

\begin{example}
Let $k=4$ and $C= \{a, b, c, d, e\}$. $19$ voters cast the following
ballots: $7$ voters approve of $\{a, b, d\}$, $4$ voters approve of $\{a,
b, e\}$, $3$ voters approve of $\{a, c, d\}$, and $5$ voters approve of
$\{a, c, e\}$. For this election both SeqPAV and seq-Phtagm\'en output
only $\{a, b, c, d\}$ (the candidates are added to the set of winners
in this order). Now, if an additional voter enters the election and
approves of only $\{c, d\}$, then both SeqPAV and seq-Phtagm\'en
output only $\{a, d, e, b\}$ (the candidates are added to the set of
winners in this order). This proves that both SeqPAV and
seq-Phragm\'en fail strong population monotonicity with population
increase. To prove that SeqPAV and seq-Phragm\'en fail strong
population monotonicity without population increase we simply add an
additional candidate $f$ to the original election and a voter that
approves of $\{f\}$. This does not make any difference and the set of
winners both with SeqPAV and seq-Phragm\'en will be again $\{a, b, c,
d\}$. Now, if this new voter decides to approve of $\{c, d, f\}$, then
both SeqPAV and seq-Phragm\'en output only $\{a, d, e, b\}$.
\end{example}

\citeauthor{mora2015eleccions}~\shortcite{mora2015eleccions} have
previously observed that seq-Phragm\'en fails the strong support
monotonicity axioms. This fact has also been discussed by
\citeauthor{2016arXiv161108826J}~\shortcite{2016arXiv161108826J} (they
use different examples from the one presented here). The previous
example is included for competeness and as a counterexample for
SeqPAV.

We study now support monotonicity for
max-Phragm\'en. \citeauthor{phragmen:pOpt}~\shortcite{phragmen:pOpt}
proved that max-Phragm\'en satisfies support monotonicity when
$|G|=1$. That proof could be easily extended to prove that max-Phragm\'en
satisfies weak SMWPI and weak SMWOPI.

\begin{theorem}
\label{th:six}  
max-Phragm\'en satisfies weak SMWPI and weak SMWOPI.
\end{theorem}

\begin{proof}[Proof]
We first prove weak SMWOPI. Consider an election $\mathcal{E}
=(\mathcal{A}, k)$, a set of winners $W$ output by max-Phragm\'en for
election $\mathcal{E}$, a non-empty subset $G$ of $W$, and a voter $i$
that does not approve of any of the candidates in $G$. Let ${\bf
  x}^{\mathrm{opt}}= (x_{i',c}^{\mathrm{opt}})_{i' \in N, c \in W}$ be
a load distribution that minimizes the maximum voter load for election
$\mathcal{E}$ and candidates subset $W$, and let $m_{\mathcal{E}}$ be
the maximum voter load for load distribution ${\bf x}^{\mathrm{opt}}$,
that is, $m_{\mathcal{E}}= \max_{i' \in N} x_{i'}^{\mathrm{opt}}$.

Observe that ${\bf x}^{\mathrm{opt}}$ is a valid, possibly
non-optimal, load distribution for election $\mathcal{E}_{i+G}$ and
candidates subset $W$. In particular, for each candidate $c$ that
belongs to $G$, since voter $i$ does not approve of $c$ in election
$\mathcal{E}$, it holds that $x_{i,c}^{\mathrm{opt}}= 0$. 

Consider now any candidates subset $W'$ of size $k$ such that $W' \cap
G= \emptyset$. Observe that for the candidates subset $W'$ the
set of valid load distributions for election $\mathcal{E}_{i+G}$ are
the same as the set of valid load distributions for election
$\mathcal{E}$. In particular, for voter $i$, the candidates for which
$x_{i,c}$ can be greater than 0 are $A_i \cap W'$ both in election
$\mathcal{E}$ and in election $\mathcal{E}_{i+G}$. It follows
immediately that the minimum maximum voter load for candidates subset
$W'$ is the same in elections $\mathcal{E}$ and $\mathcal{E}_{i+G}$.

Since the minimum maximum voter load for the candidates subset $W$
does not increase in election $\mathcal{E}_{\Delta G}$ with respect to
election $\mathcal{E}$ and, for each candidates subset $W'$ such that
$W' \cap G= \emptyset$ the minimum maximum voter load for the
candidates subset $W'$ is the same in elections $\mathcal{E}_{i+G}$
and $\mathcal{E}$, it follows that $W$ or some candidates subset that
contains some of the candidates in $G$ must be output by
max-Phragm\'en for election $\mathcal{E}_{i+G}$. Further, if for all
the set of winners $W$ output by max-Phragm\'en for election
$\mathcal{E}$ it holds that $G \subseteq W$, then for each candidates
subset $W'$ such that $W' \cap G= \emptyset$ the minimum maximum voter
load for the candidates subset $W'$ is strictly greater than
$m_{\mathcal{E}}$, and therefore, it cannot be a set of winners for
election $\mathcal{E}_{i+G}$.

The proof for weak SMWPI follows from the facts that max-Phragm\'en
satisfies weak SMWOPI, and that for any election $\mathcal{E}$ the
sets of winners output by max-Phragm\'en do not change if we add a
voter to the election that does not approve of any candidate.
\end{proof}

However, the following example proves that max-Phragm\'en fails both
strong SMWPI and strong SMWOPI.

\begin{example}
Let $k=6$ and $C=\{a, b, c_1, \ldots, c_5\}$. $18$ voters cast the
following ballots: $13$ voters approve of $\{c_1, \ldots, c_5\}$, $2$
voters approve of $\{a, b\}$, $2$ voters approve of $\{a\}$, and $1$
voter approves of $\{b\}$. For this election max-Phragm\'en outputs
only one set of winners: $\{a, c_1, \ldots, c_5\}$. The minimum
maximum load for this election is achieved as follows: for each voter
$i$ that approves of $\{c_1, \ldots, c_5\}$ and each candidate $c$ in
$\{c_1, \ldots, c_5\}$ we have $x_{i,c}= \frac{1}{13}$, and for each
voter $i'$ that approves of $a$ we have $x_{i',a}= \frac{1}{4}$. Then,
the load of the voters that approve of $\{c_1, \ldots, c_5\}$ is
$\frac{5}{13}$ and the load of the voters that approve of $a$ is
$\frac{1}{4}$. The maximal voter load for this example is therefore
$\frac{5}{13}$. Now, if a new voter enters the election and approves
of precisely $\{a, c_1, \ldots, c_5\}$, then the sets of winners
outputted by max-Phragm\'en consist of $\{a, b\}$ plus $4$ candidates
from $\{c_1, \ldots, c_5\}$. In this case the minimum maximum voter
load is achieved by assigning again $x_{i,c}= \frac{1}{13}$ for each
voter $i$ that approves of $\{c_1, \ldots, c_5\}$ and each candidate
$c$ in $\{c_1, \ldots, c_5\}$, assigning $x_{i,a}= \frac{1}{3}$ to the
new voter and the voters that approve of $\{a\}$, and assigning
$x_{i,b}= \frac{1}{3}$ to all the voters that approve of candidate
$b$. This leads to a maximum voter load of $\frac{1}{3}$. Observe that
in this case the minimum maximum voter load for the set $\{a, c_1,
\ldots, c_5\}$ would be obtained by $x_{i,c}= \frac{1}{14}$ for each
voter $i$ that approves of $\{c_1, \ldots, c_5\}$, and also for the
new voter, which leads to a maximum voter load of $\frac{5}{14}$,
greater than $\frac{1}{3}$. This example proves that max-Phragm\'en
fails strong support monotonicity with population increase.

To prove that max-Phragm\'en fails strong support monotonicity
without population increase we simply add an additional candidate $d$ to
the original election and a voter that approves of $\{d\}$. This does
not make any difference and the set of winners will be again $\{a,
c_1, \ldots, c_5\}$. Now, if this new voter decides to approve of
$\{a, c_1, \ldots, c_5, d\}$, then the sets of winners outputted by
max-Phragm\'en consist of $\{a, b\}$ plus $4$ candidates from
$\{c_1, \ldots, c_5\}$.
\end{example}

\section{Committee monotonicity}
\label{sec:comMon}

We turn now to discuss briefly committee monotonicity. The following
definition, due to \citeauthor{elkind:scw17}~\shortcite{elkind:scw17},
was given in the context of multi-winner voting rules that make use of
ranked ballots, but it can also be directly used for approval-based
multi-winner voting rules.

\begin{definition}
We say that a voting rule $R$ satisfies committee monotonicity if for
every election $\mathcal{E}=(N, C, \mathcal{A}, k)$, with $k \in \{1,
\ldots, |C|-1\}$, the following conditions hold:

\begin{enumerate}

\item
for each $W$ in $R(\mathcal{E}= (N, C, \mathcal{A}, k))$ there exists
a $W'$ in $R(N, C, \mathcal{A}, k+1)$ such that $W \subseteq W'$, and

\item

for each $W$ in $R(N, C, \mathcal{A}, k+1)$ there exists a $W'$ in
$R(\mathcal{E}= (N, C, \mathcal{A}, k))$ such that $W' \subseteq W$.

\end{enumerate}
 
\end{definition}

It is generally believed that committee monotonicity is a desirable
axiom for scenarios of type
excellence~\cite{faliszewski2017multiwinner,elkind:scw17,barbera2008choose}.

It is easy to see that committee monotonicity is satisfied by the
rules that consist of an iterative algorithm such that at each
iteration the candidate that is added to the set of winners does not
depend on the target committee size. This holds for AV, SAV, seqPAV,
and seq-Phragm\'en.

For the remaining rules it is easy to find counterexamples where they
fail committee monotonicity. We start with MAV.

\begin{example}
\label{ex:cm-mav}  
Let $C=\{a, b, c\}$. $4$ voters cast the following ballots: $1$ voter
approves of $\{b\}$, $1$ voter approves of $\{c\}$, $1$ voter approves
of $\{a,b\}$, and $1$ voter approves of $\{a,c\}$. For this set of
candidates and this ballot profile, for $k=1$ MAV outputs only $\{a\}$
(with a maximum Hamming distance of $2$). For $k=2$, MAV outputs only
$\{b,c\}$ (also with a maximum Hamming distance of $2$).
\end{example}

\citeauthor{thiele:pav-rav}~\shortcite{thiele:pav-rav}
and~\citeauthor{mora2015eleccions}~\shortcite{mora2015eleccions} have
already proved that PAV and max-Phragm\'en, respectively, fail
committee monotonicity. We give here an example that shows that
both CC and Monroe fail committee monotonicity.

\begin{example}
\label{ex:cm-monroecc}  
Let $C=\{a, b, c\}$. $10$ voters cast the following ballots: $3$
voters approve of $\{a, b\}$, $3$ voters approve of $\{a, c\}$, $2$
voters approve of $\{b\}$ and $2$ voters approve of $\{c\}$. For this
set of candidates and this ballot profile, for $k=1$ both CC and
Monroe output only $\{a\}$. For $k=2$, both CC and Monroe output only
$\{b,c\}$.
\end{example}

\section{Compatibility of axioms}
\label{sec:comp}

In many applications it would be interesting to use voting rules that
satisfy both support monotonicity and representation axioms. While all
the voting rules that we have analyzed that satisfy PJR (or EJR) also
satisfy the weak support monotonicity axioms, the situation changes
when we require the strong axioms. In particular, none of the rules
analyzed that satisfy PJR also satisfy strong SMWOPI, and only PAV
(which has the additional difficulty of being NP-hard to compute)
satisfies strong SMWPI. Whether it is possible to develop a voting
rule that satisfies strong SMWOPI and PJR at the same time is left
open.

In contrast, we can formally prove that perfect representation (PR) is
incompatible both with strong SMWPI and with committee
monotonicity. We review first the definition of PR due to
\citeauthor{pjr-aaai}~\shortcite{pjr-aaai}.

\begin{definition}

{\bf Perfect representation} (PR) Consider a ballot profile
$\mathcal{A}$ over a candidate set $C$, and a target committee size
$k$, $k \leq |C|$, such that $k$ divides $n$. We say that a set of
candidates $W$, $|W| = k$, provides perfect representation (PR) for
$(\mathcal{A}, k)$ if it is possible to partition the set of voters in
$k$ pairwise disjoint subsets $N_1, \dots, N_k$ of size $\frac{n}{k}$
each, such that each candidate $w$ in $W$ can be assigned to one (and
only one) different subset $N_i$ so that for all pairs $(w, N_i)$ all
the voters in $N_i$ approve of their assigned candidate $w$. We say
that an approval-based voting rule satisfies perfect representation
(PR) if for every election $(\mathcal{A}, k)$ it does not output any
winning set of candidates $W$ that does not provide PR for
$(\mathcal{A}, k)$ if at least one set of candidates $W'$ that
provides PR for $(\mathcal{A}, k)$ exists.

\end{definition}

\begin{theorem}
No rule can satisfy PR and strong SMWPI at the same time.
\label{theo:PR_SMWPI}
\end{theorem}

\begin{proof}[Proof]
Consider the following election. Let $k=3$ and $C= \{c_1, \ldots,
c_5\}$. 12 voters cast the following ballots: 2 voters approve of
$\{c_1, c_4\}$, 2 voters approve of $\{c_1, c_5\}$, 3 voters approve
of $\{c_2, c_4\}$, one voter approves of $\{c_2, c_5\}$, 2 voters
approve of $\{c_3, c_5\}$, and 2 voters approve of $\{c_3\}$. For this
election any voting rule that satisfies PR has to output $\{c_1, c_2,
c_3\}$. Now, suppose that 3 new voters enter the election, and that
all these new voters approve of $\{c_1, c_3\}$. For this extended
election a voting rule that satisfies PR has to output only $\{c_3,
c_4, c_5\}$.
\end{proof}

There is an apparent contradiction between this theorem and
Table~\ref{tab:comparisonSummary} because
Table~\ref{tab:comparisonSummary} says that CC satisfies both PR and
strong SMWPI. The reason for this apparent contradiction is that, as
explained in Footnote~g, CC satisfies PR {\bf only} if ties are broken
in favour of the sets of candidates that provide PR. The example of
Theorem~\ref{theo:PR_SMWPI} illustrates this. For the initial election
CC outputs $\{c_1, c_2, c_3\}$ and $\{c_3, c_4, c_5\}$. However, if
ties are broken in favour of the sets of candidates that provide PR,
then CC (with this tie-breaking rule) will output only $\{c_1, c_2,
c_3\}$. Now, after adding 3 new voters that approve of $\{c_1, c_3\}$,
strong SMWPI requires that both $c_1$ and $c_3$ are in the set of
winners while PR requires that the set of winners is $\{c_3, c_4,
c_5\}$.

Whether strong SMWOPI and PR are compatible axioms is unclear. We
observe that if a certain candidates subset $W$ provides PR for a
certain election $\mathcal{E}= (\mathcal{A}, k)$, then for each
non-empty candidates subset $G$ of $W$, and for each voter $i$ such
that $A_i \cap G= \emptyset$, it holds that $W$ also provides PR for
$\mathcal{E}_{i+G}$. It is enough to observe that the same assignment
between candidates and voters that works for $\mathcal{E}$ will also
work for $\mathcal{E}_{i+G}$. In particular if voter $i$ approved of
her assigned candidate $w$ in election $\mathcal{E}$, this means that
$w \in A_i$, and therefore voter $i$ approves of $w$ also in election
$\mathcal{E}_{i+G}$.

\begin{theorem}
No rule can satisfy PR and committee monotonicity at the same time.
\end{theorem}

\begin{proof}[Proof]
Consider the following election. Let $C= \{c_1, \ldots, c_5\}$. 6
voters cast the following ballots. For $i= 1$ to $3$, and for $j=1$ to
$2$, one voter approves of $\{c_i, c_{3+j}\}$. If the target committee
size is 2, a voting rule that satisfies PR has to output only $\{c_4,
c_5\}$, but if the target committee size is 3, a voting rule that
satisfies PR has to output only $\{c_1, c_2, c_3\}$.
\end{proof}

\begin{table}[htb]
  \begin{tabular}{|l|l|l|l|} \hline
    & {\bf SMWPI} & {\bf SMWOPI} & {\bf Com. Mon.} \\ \hline
{\bf JR} & Str. & Wk.(Str.?) & Yes \\ \hline    
{\bf PJR} & Str. & Wk.(Str.?) & Yes \\ \hline    
{\bf EJR} & Str. & Wk.(Str.?) & ? \\ \hline
{\bf PR} & Wk. & Wk.(Str.?) & No \\ \hline    
  \end{tabular}  
\caption{\label{tab:compat}Summary of results on compatibility between
  representation and monotonicity axioms}
\end{table}
  
Table~\ref{tab:compat} summarises the results that we have found in
this paper with respect to the compatibility between representation
and monotonicity axioms. Of course, the rules that we have found
before that satisfy a certain monotonicity axiom and a certain
representation axiom at the same time prove that such axioms are
compatible. In the table, ``Str.'' means that the strong version of
the support monotonicity axiom is compatible with the corresponding
representation axiom, ``Wk.'' means that the weak version of the
support monotonicity axiom is compatible with the corresponding
representation axiom but that the strong version of the support
monotonicity axiom and the corresponding representation axiom are
incompatible, and ``Wk.(Str.?)'' means that the weak version of the
support monotonicity axiom is compatible with the corresponding
representation axiom but that we do not know whether the the strong
version of the support monotonicity axiom and the corresponding
representation axiom are compatible.

\section{Related Work}
\label{sec:relw}

There exists some previous work that consider support monotonicity in
the context of approval-based multi-winner
elections. \citeauthor{2017arXiv170807580A}~\shortcite{2017arXiv170807580A}
propose several notions of monotonicity for weak preferences (ties
between candidates are allowed), and then they consider the
restriction of such notions to approval ballots. The strongest
monotonicity axiom that they propose when restricted to approval
ballots corresponds essentially to the candidate monotonicity axiom
that we have presented before (they also call this axiom candidate
monotonicity).

\citeauthor{lackner2018approval}~\shortcite{lackner2018approval}
consider a different (and much more stronger) notion of support
monotonicity, defined for a subclass of the approval-based multi-winner
voting rules called ABC ranking rules. They use this axiom to
characterize a subset of the ABC ranking rules that they call
Dissatisfaction Counting Rules (AV belongs to this class of rules).
In particular, a necessary (but not sufficient) condition to satisfy
such axiom is that if a voter $i$ that is already in the election adds
a candidate $w$ that was already in the set of winners $W$, then $W$
must still be a set of winners (possibly tied with other sets of
winners). It follows that this axiom is strictly stronger than strong
SMWOPI.

According to
\citeauthor{2016arXiv161108826J}~\shortcite{2016arXiv161108826J},
Phragm\'en also studied support monotonicity in the approval-based
rules that he proposed. In particular, Phragm\'en proved that
max-Phragm\'en and seq-Phragm\'en satisfy support monotonicity when
only one candidate increases her support, either because a voter
already in the election adds such candidate to her ballot (candidate
monotonicity) or because a new voter enters the election and approves
of only such candidate.

\citeauthor{mora2015eleccions}~\shortcite{mora2015eleccions} and
\citeauthor{2016arXiv161108826J}~\shortcite{2016arXiv161108826J} have
recently extended the study of the monotonicity properties of
seq-Phragm\'en. In particular, they give examples that prove that
seq-Phragm\'en fails both strong SMWPI and strong SMWOPI. We stress
the following differences between the works of
\citeauthor{mora2015eleccions}~\shortcite{mora2015eleccions} and
\citeauthor{2016arXiv161108826J}~\shortcite{2016arXiv161108826J} and
ours: 1) they do not consider weak SMWPI and weak SMWOPI; and 2) they
do not formalize the strong axioms (they only give examples that show
that seq-Phragm\'en fails them).

There is also some relation between our work and the work of
\citeauthor{peters2018proportionality}~\shortcite{peters2018proportionality}.
\citeauthor{peters2018proportionality}~\shortcite{peters2018proportionality}
defines an axiom that they call strategyproofness that is similar to
SMWOPI. For a rule $f$ that satisfies this axiom it cannot happen that
$W' \cap (A_i \cup G) \subsetneq W \cap (A_i \cup G)$, were $W$ is the
output of rule $f$ for election $\mathcal{E}= (\mathcal{A}, k)$ and
$W'$ is the output of rule $f$ for election $\mathcal{E}_{i+G}$ (they
assume that the rules are resolute). As we have already said, this
axiom is similar to SMWOPI although neither strong SMWOPI implies
strategyproofness nor strategyproofness implies weak
SMWOPI. \citeauthor{peters2018proportionality}~\shortcite{peters2018proportionality}
prove that strategyproofness is not compatible with a representation
axiom that is even weaker than JR using SAT solvers.

\section{Conclusions and future work} 

In this paper we have complemented previous work on the axiomatic
study of multi-winner voting rules with the study of monotonicity
axioms for rules that use approval ballots. Our results show that
support monotonicity in approval-based multi-winner voting rules is
trickier than it may seem at first glance. While the weak support
monotonicity axioms are satisfied in almost all the cases analyzed in
this study (only Monroe fails one of these) the situation changes
completely when we look to the strong axioms. Of the 7 rules analyzed
only 4 satisfy strong SMWPI and only 2 satisfy strong SMWOPI.

We have also presented some results related to the compatibility
between representation and monotonicity axioms. First, we have proved
that PR is incompatible both with strong SMWPI and with committee
monotonicity. Our results also show that EJR and PJR are compatible
with strong SMWPI and weak SMWOPI (in particular, PAV satisfies all
these axioms). Our incompatibility results are mostly of theoretical
interest because PR rules are NP-hard to compute~\cite{pjr-aaai} and
therefore of little practical use. However, we believe that these
results are interesting because they illustrate the existence of a
certain conflict between representation and monotonicity axioms. With
respect to the compatibility between EJR and PJR with monotonicity
axioms, several interesting open questions remain open. First of all,
the only rule that we have found that satisfies EJR (or PJR) and
strong SMWPI is PAV, which is known to be NP-hard to
compute. Therefore, it would be very interesting to find a rule that
satisfies EJR (or PJR) and strong SMWPI but can be computed in
polynomial time. Very recently,
\citeauthor{aziz2018complexity}~\shortcite{aziz2018complexity} have
identified a set of voting rules that satisfy EJR and can be computed
in polynomial time. The study of these rules could be interesting in
this regard.

In the second place, it is also open whether EJR and PJR are
compatible with strong SMWOPI. Other open issues would be to find a
rule that satisfies both EJR and committee monotonicity and to find a
rule that satisfies strong SMWOPI and PR. The similarity between the
strategyproofness axiom proposed by
\citeauthor{peters2018proportionality}~\shortcite{peters2018proportionality}
and SMWOPI makes us think that the use of SAT solvers could be a
possible approach to address these research questions.

We have also studied the relevance of our axioms to several types of
scenarios. We have found that our support monotonicity axioms fit well
with all the cases studied except in the case of the second type of
proportional representation discussed by
\citeauthor{pjr-aaai}~\shortcite{pjr-aaai}, where the goal is to
satisfy large cohesive groups according to their size. Therefore, it
would be interesting to define an additional support monotonicity
axiom that fits well in this scenario and is compatible with EJR (which
is oriented to this type of proportional representation). The
development of an adapted version of the notion of noncrossing
monotonicity proposed by
\citeauthor{elkind:scw17}~\shortcite{elkind:scw17} for ranked ballots
could also be a line of continuation of this work.

\newpage

\appendix

\section{Examples of rules that satisfy SMWPI or SMWOPI when $|G|=1$ but fail the corresponding weak axiom}
\label{ap:examples}

In this section we present two rules that satisfy SMWPI or SMWOPI when
$|G|=1$ but fail the corresponding weak axiom. They are probably
useless in practice but show that such situation is possible.

\subsection{A rule that satisfies SMWPI when $|G|=1$ but fails weak SMWPI}

The following rule trivially fails weak SMWPI but satisfies SMWPI when
$|G|=1$:

\begin{displaymath}
r^{\lnot \textrm{weakSMWPI}}= \underset{W \subseteq C:
  |W| = k}{\textrm{argmax}} \ \sum_{c \in W} |\{i: A_i= \{c\}\}|
- \sum_{c \in W} |\{i: c \in A_i, |A_i| \geq 2\}|
\end{displaymath}  

\subsection{A rule that satisfies SMWOPI when $|G|=1$ but fails weak SMWOPI}

A rule that satisfies SMWOPI when $|G|=1$ but fails weak SMWOPI can be
obtained by tweaking AV.

Given a candidates set $C$, a permutation $\sigma_C: C \rightarrow C$
over such candidates set, and a subset $A$ of $C$, we say that
$\sigma_C(A) = \{c \in C: \exists c' \in A, c= \sigma_C(c')\}$. Given
an election $\mathcal{E}= (N, C, \mathcal{A}= (A_1, \ldots, A_n), k)$,
a permutation $\sigma_N: N \rightarrow N$ over $N$, and a permutation
$\sigma_C: C \rightarrow C$ over $C$, we say that
$\sigma_N(\mathcal{E})= (N, C, \mathcal{A}'= (A_{\sigma_N(1)}, \ldots,
A_{\sigma_N(n)}), k)$, and that $\sigma_C(\mathcal{E})= (N, C,
\mathcal{A}''= (\sigma_C(A_1), \ldots, \sigma_C(A_n)), k)$.

The rule $\textrm{AV}^{\lnot \textrm{weakSMWOPI}}$ outputs exactly the
same sets of winners as AV except in the particular case that in the
election there are 4 voters, 4 candidates, and $k= 2$. In such
particular case, the election can be represented as $\mathcal{E}= (N=
\{1,2,3,4\}, C=\{c_1, c_2, c_3, c_4\}, \mathcal{A}= (A_1, A_2, A_3,
A_4), k=2)$.

Then,

\begin{displaymath}
  \begin{array}{rl}
    \mathcal{E}_1 &=
(N,C,(\{c_1,c_2\},\{c_3,c_4\},
\{c_1,c_3,c_4\},\{c_2,c_3,c_4\}),k) \\    
\textrm{AV}^{\lnot \textrm{weakSMWOPI}} (\mathcal{E}_1) &=
\{\{c_1,c_2\}, \{c_3,c_4\}\} \\
    \mathcal{E}_2 &=
(N,C,(\{c_1,c_2\},\{c_1,c_2,c_3,c_4\},
\{c_1,c_3,c_4\},\{c_2,c_3,c_4\}),k) \\
\textrm{AV}^{\lnot
  \textrm{weakSMWOPI}} (\mathcal{E}_2) &=
\{\{c_3,c_4\}\} \\
\end{array}
\end{displaymath}

Observe that this automatically implies that $\textrm{AV}^{\lnot
  \textrm{weakSMWOPI}}$ fails weak SMWOPI, because if the voter that
approves of $\{c_3,c_4\}$ in the $\mathcal{E}_1$ decides to approve of
$\{c_1,c_2,c_3,c_4\}$ neither $c_1$ nor $c_2$ can be in the set of
winners.

We impose that for any permutation $\sigma_N: N \rightarrow N$ over
the voters, and any permutation $\sigma_C: C \rightarrow C$ over the
candidates, $W$ is a set of winners for
$\sigma_C(\sigma_N(\mathcal{E}_1))$ under rule $\textrm{AV}^{\lnot
  \textrm{weakSMWOPI}}$ if and only if $\sigma_C(W)$ is a set of
winners for election $\mathcal{E}_1$. Also, $W$ is a set of winners
for $\sigma_C(\sigma_N(\mathcal{E}_2))$ under rule $\textrm{AV}^{\lnot
  \textrm{weakSMWOPI}}$ if and only if $\sigma_C(W)$ is a set of
winners for election $\mathcal{E}_2$.

In all the remaining cases, the rule $\textrm{AV}^{\lnot
  \textrm{weakSMWOPI}}$ outputs the same sets of winners as AV. It can
be shown that $\textrm{AV}^{\lnot \textrm{weakSMWOPI}}$ satisfies
SMWOPI when $|G|=1$.

\section{Results related to perfect representation}

We provide here the proofs and counterexamples related to PR for the
rules studied in this paper and that have not been studied
elsewhere. First we give examples that shows that AV, SAV, MAV, and
seqPAV fail PR and then we prove that CC satisfies
PR if combined with the appropriate tie-breaking rule.

\begin{example}
\label{ex:pr-av}
Let $k= 3$ and $C= \{a_1, a_2, a_3, b_1, b_2, b_3\}$. $3$ voters cast
the following ballots: $2$ voters approve of $\{a_1, a_2, a_3\}$ and $1$
voter approves of $\{b_1, b_2, b_3\}$. In this example both AV and SAV
output $\{a_1, a_2, a_3\}$. However, all the candidates subsets that
provide PR for this example have to include one candidate from $\{b_1,
b_2, b_3\}$.
\end{example}

MAV also fails PR as the following example shows.

\begin{example}
\label{ex:pr-mav}  
Let $k= 3$ and $C= \{a_1, a_2, b_1, b_2, b_3\}$. $3$ voters cast the
following ballots: $2$ voters approve of $\{a_1, a_2\}$ and $1$ voter
approves of $\{b_1, b_2, b_3\}$. In this example the sets of winners
outputted by MAV contain $1$ candidate from $\{a_1, a_2\}$ and $2$
candidates from $\{b_1, b_2, b_3\}$. However, all the candidates
subsets that provide PR for this example have to include $\{a_1,
a_2\}$ and one candidate from $\{b_1, b_2, b_3\}$.
\end{example}

Finally, here is an example for SeqPAV.

\begin{example}
\label{ex:pr-seqpav}  
Let $k=2$ and $C= \{a, b, c\}$. $6$ voters voters cast the following
ballots: $2$ voters approve of $\{a, b\}$, $2$ voters approve of $\{a,
c\}$, $1$ voter approves of $\{b\}$ and $1$ voter approves of
$\{c\}$. The only candidates subset that provides PR for this example
is $\{b, c\}$. However, SeqPAV adds candidate $a$ (the most approved
of one) to the set of winners in the first iteration.
\end{example}

Consider now the case of CC. It is evident that for each election
$\mathcal{E}= (\mathcal{A}, k)$ in which candidates subsets that
provide PR exist, for each candidates subset $W$ that provides PR for
election $\mathcal{E}$ it holds that all the voters that participate
in the election approve of some of the candidates in $W$. It follows
immediately that $W$ would have misrepresentation 0 according to the
rules in CC, and therefore that $W$ is outputted by CC for election
$\mathcal{E}$. It may happen, however, that CC outputs sets of winners
that do not provide PR even if sets of winners that provide PR exist.

\begin{example}
\label{ex:pr-cc}  
We consider again the election in example~\ref{ex:pr-av}. One of the
sets of winners that CC outputs for this election is $\{a_1, b_1,
b_2\}$. However, all the candidates subsets that provide PR for this
example have to include two candidates from $\{a_1, a_2, a_3\}$.
\end{example}

In summary, CC satisfies PR only if ties are broken always in favour of
the candidates subsets that provide PR.

\bibliographystyle{ACM-Reference-Format}  
\bibliography{dhondt}

\end{document}